\documentclass[11pt, a4paper]{article}

\usepackage{amsmath}
\usepackage{amssymb}
\usepackage{amsfonts}
\usepackage{amsthm}
\usepackage{graphicx}
\usepackage{fullpage}
\usepackage{hyperref}

\newtheorem{theorem}{Theorem}[section]
\newtheorem{lemma}[theorem]{Lemma}

\newtheorem{corollary}[theorem]{Corollary}
\newtheorem{claim}[theorem]{Claim}
\newtheorem{observation}[theorem]{Observation}

\newtheorem{remark}[theorem]{Remark}
\newtheorem{example}[theorem]{Example}

\usepackage{tabularx}
\usepackage{environ}

\makeatletter
\newcommand{\problemtitle}[1]{\gdef\@problemtitle{#1}}
\newcommand{\probleminput}[1]{\gdef\@probleminput{#1}}
\newcommand{\problemoutput}[1]{\gdef\@problemoutput{#1}}
\NewEnviron{problem}{
  \problemtitle{}\probleminput{}\problemoutput{}
  \BODY
  \par\addvspace{.5\baselineskip}
  \noindent{
    \framebox[.98\textwidth][c]{
        \begin{tabularx}{\textwidth}{@{\hspace{\parindent}} l X c}
            \multicolumn{2}{@{\hspace{\parindent}}l}{\textsc{\@problemtitle}} \\
            \textbf{Input:} & \@probleminput \\
            \textbf{Question:} & \@problemoutput
        \end{tabularx}
        }
      }
  \par\addvspace{.5\baselineskip}
}
\makeatother

\usepackage[linesnumbered,noend,ruled,vlined]{algorithm2e}
\SetKwInput{KwInput}{Input}     
\SetKwInput{KwOutput}{Output}   

\usepackage{algorithmic}

\title{Optimal General Factor Problem and Jump System Intersection}

\author{Yusuke Kobayashi
\thanks{Kyoto University.
E-mail: yusuke@kurims.kyoto-u.ac.jp}
}

\date{}

\begin{document}
\maketitle
\begin{abstract}
In the optimal general factor problem, 
given a graph $G=(V, E)$ and a set $B(v) \subseteq \mathbb Z$ of integers for each $v \in V$, 
we seek for an edge subset $F$ of maximum cardinality subject to $d_F(v) \in B(v)$ for $v \in V$, 
where $d_F(v)$ denotes the number of edges in $F$ incident to $v$. 
A recent crucial work by Dudycz and Paluch shows that 
this problem can be solved in polynomial time if each $B(v)$ has no gap of length more than one. 
While their algorithm is very simple, its correctness proof is quite complicated. 
In this paper, we formulate the optimal general factor problem as the jump system intersection, and 
reveal when the algorithm by Dudycz and Paluch can be applied to this abstract form of the problem. 
By using this abstraction,  
we give another correctness proof of the algorithm, 
which is simpler than the original one. 
We also extend our result to the valuated case. 
\end{abstract}

\section{Introduction}

\subsection{General Factor Problem}

Matching in graphs is one of the most well-studied topics in combinatorial optimization. 
Since a maximum matching algorithm was proposed by Edmonds~\cite{edmonds_1965} in 1960s, 
a lot of generalizations of the matching problem have been proposed and studied in the literature.  
Among them, we focus on the \emph{general factor problem}, which contains several important problems as special cases. 
In the general factor problem (or also called the \emph{$B$-factor problem}), 
we are given a graph $G=(V, E)$ and a set $B(v) \subseteq \mathbb Z$ of integers for each $v \in V$. 
The objective is to find an edge subset $F \subseteq E$ such that $d_F(v) \in B(v)$ for any $v \in V$ if it exists, 
where $d_F(v)$ denotes the number of edges in $F$ incident to $v$. 
Such an edge set is called a \emph{$B$-factor}. 

Since the general factor problem is NP-hard in general (e.g.~it contains the $3$-edge-coloring problem~\cite{lovasz72}), 
polynomially solvable special cases have attracted attention. 
A $B$-factor amounts to a perfect matching if $B(v) = \{1\}$ for each $v \in V$, and 
it is called a \emph{$b$-factor} if $B(v) = \{b(v)\}$ for each $v \in V$, where $b \colon V \to \mathbb Z$. 
For $a, b \colon V \to \mathbb Z$, 
if $B(v) = \{a(v), a(v)+1, a(v) + 2, \dots , b(v) -1, b(v)\}$ (resp,~$B(v) = \{a(v), a(v)+2, a(v) + 4, \dots , b(v) -2, b(v)\}$) for $v \in V$, then
a $B$-factor is called an \emph{$(a, b)$-factor} (resp.~an \emph{$(a, b)$-parity factor}).
It is well-known that, in the above cases, 
we can find a $B$-factor in polynomial time by using a maximum matching algorithm; see~\cite{lovasz72} and \cite[Section 35]{lexbook}. 
Note that the parity constraint can be dealt with by adding $\frac{1}{2} (b(v) - a(v))$ self-loops to each $v \in V$ and modifying $B(v)$. 
Another special case is the \emph{antifactor problem}, in which 
$B(v) = \{0, 1, 2, \dots , d_E(v)\} \setminus \{ \alpha_v \}$ for some $\alpha_v \in \{0, 1, 2, \dots , d_E(v)\}$, that is, exactly one value is forbidden for each $v \in V$. 
Graphs with  an antifactor were characterized by Lov\'asz~\cite{lovasz73}. 
The \emph{edge-and-triangle partitioning problem} is to cover all the vertices in a graph by edges and triangles that are mutually disjoint, 
which can be easily reduced to the general factor problem with $B(v) = \{1\}, \{0, 2\}$, or $\{0, 2, 3\}$. 
The edge-and-triangle partitioning problem is known to be solvable in polynomial time~\cite{CORNUEJOLS1982139}. 

All the above polynomially solvable cases have a property that each $B(v)$ has no \emph{gap} of length more than one. 
Here,  $B(v) \subseteq \mathbb Z$ is said to have \emph{a gap of length $p$} if there exists $\alpha \in B(v)$ such that 
$\alpha + 1, \alpha + 2, \dots , \alpha +p \not\in B(v)$ and $\alpha + p + 1 \in B(v)$. 
It turns out that this is a key property to design a polynomial-time algorithm.
Indeed, Cornu\'ejols~\cite{CORNUEJOLS1988185} gave a polynomial-time algorithm for the general factor problem with this property
and Seb\H{o}~\cite{10.1006/jctb.1993.1035} gave a good characterization.

An optimization variant of the general factor problem has also attracted attention, which we call the \emph{optimal general factor problem} (or the \emph{optimal general matching problem}).  
In the problem, given a graph $G=(V, E)$ and a set $B(v) \subseteq \mathbb Z$ of integers for each $v \in V$, 
we seek for a $B$-factor of maximum cardinality. 
It is the maximum matching problem if $B(v) = \{0, 1\}$, and 
is the maximum $b$-matching problem if $B(v) = \{0, 1, \dots , b(v)\}$, both of which can be solved in polynomial time. 
In the same way as the search problem described above, we can find a maximum $(a, b)$-factor (or $(a, b)$-parity factor) in polynomial time; see~\cite[Section 35]{lexbook}. 
The optimization variant of the edge-and-triangle partitioning problem was studied with the name of the \emph{simplex matching problem}, 
and a polynomial-time algorithm was designed for this problem~\cite{termback}; see also~\cite{pap}.

Recently, Dudycz and Paluch~\cite{DudyczP17} showed that 
the optimal general factor problem can be solved in polynomial time if each $B(v)$ has no gap of length more than one. 
This is definitely a crucial result in this area, because it is a generalization of all the above results. 
While their algorithm is very simple, its correctness proof is quite complicated.

\subsection{Jump System Intersection}

In this paper, we introduce an abstract form of the optimal general factor problem
by using the concept of {\em jump systems} introduced by Bouchet and Cunningham~\cite{Bouchet1995DeltaMatroidsJS} (see also~\cite{Kabadi2005,lovasz97}). 
Let $V$ be a finite set. For $x, y \in \mathbb{Z}^V$, we say that $s \in \mathbb{Z}^V$ is an {\em $(x, y)$-step} if 
$\|s\|_1 = 1$ and $\| (x+s) - y\|_1 = \|x-y\|_1 - 1$. 
A non-empty subset $J \subseteq \mathbb{Z}^V$ is called a {\em jump system} if it satisfies the following property: 
\begin{description}
\item[(JUMP)]
For any $x, y \in J$ and for any $(x, y)$-step $s$, 
either $x + s \in J$ or there exists an $(x+s, y)$-step $t$ such that $x+s+t \in J$. 
\end{description}
Typical examples of jump systems include 
matroids, delta-matroids, integral polymatroids (or submodular systems~\cite{fujishige}), and degree sequences of subgraphs.
When $J \subseteq \mathbb Z$ is one-dimensional, one can see that 
$J$ is a jump system if and only if it has no gap of length more than one. 
One can also see that the direct product of one-dimensional jump systems is also a jump system. 
We consider the optimization problem over the intersection of two jump systems, 
where one is the direct product of one-dimensional jump systems.

\begin{problem}
\problemtitle{\sc Jump System Intersection}
\probleminput{A jump system $J \subseteq \mathbb Z^V$, a finite one-dimensional jump system $B(v) \subseteq \mathbb{Z}$ for each $v \in V$, and a vector $c \in \mathbb{Z}^V$. }
\problemoutput{Find a vector $x \in J \cap B$ maximizing $c^\top x$, where $B \subseteq \mathbb Z^V$ is the direct product of $B(v)$'s. }
\end{problem}

If $J$ consists of degree sequences of subgraphs, i.e., $J = \{d_F \in \mathbb Z^V \mid F \subseteq E\}$, and $c(v) = 1$ for $v \in V$, then 
the problem amounts to the optimal general factor problem, 
which can be solved in polynomial time~\cite{DudyczP17}. 
On the other hand, 
if $J$ is a $2$-polymatroid and $B(v) = \{0, 2\}$ for each $v \in V$, then the problem amounts to the {\em matroid matching problem}~\cite{lovaszMatroidMatching} 
or the {\em matroid parity problem}~\cite{Lawler76}. 
This implies that the problem cannot be solved in polynomial time if $J$ is given as a membership oracle~\cite{JensenKorte,lovasz80}; see also~\cite{LovaszPlummer}. 

A similar problem is to determine whether the intersection of two jump systems $J_1$ and $J_2$ is empty or not, which is also hard in general. This problem was studied in~\cite{lovasz97} as a membership problem of   $J_1 - J_2 := \{x-y \mid x \in J_1,\, y \in J_2\}$, because $J_1 \cap J_2 \neq \emptyset$ if and only if ${\bf 0} \in J_1 - J_2$.

\subsection{Our Contribution: Jump System with SBO Property}

A natural question is why the optimal general factor problem can be solved efficiently, while the general setting of {\sc Jump System Intersection} is hard. 
In this paper, we answer this question by revealing the properties of $J$ that are essential in the argument in~\cite{DudyczP17}.

For a positive integer $\ell$, we denote $\{1, 2, \dots , \ell\}$ by $[\ell]$. 
For  $x, y \in \mathbb Z^V$, we say that a multiset $\{p_1, \dots , p_\ell \}$ of vectors  is a {\em $2$-step decomposition of $y-x$} 
if $p_i \in \mathbb Z^{V}$ and  $\|p_i\|_1 =2$ for each $i \in [\ell]$, $\|y-x\|_1 = 2 \ell$, and $y-x = \sum_{i \in [\ell]} p_i$. 
A non-empty subset $J \subseteq \mathbb{Z}^V$ is called a {\em jump system with SBO property}\footnote{SBO stands for {\em strongly base orderable} (see Example~\ref{ex:01}).} if it satisfies the following property:

\begin{description}
\item[(SBO-JUMP)]
For any $x, y \in J$,  
there exists a $2$-step decomposition  $\{p_1, \dots , p_\ell \}$ of $y-x$ such that $x + \sum_{i \in I} p_i \in J$ for any $I \subseteq [\ell]$.
\end{description}
We can see that (SBO-JUMP) implies (JUMP). 
To see this,  for given $x, y \in J$, suppose that  
there exist vectors $p_1, \dots , p_\ell \in \mathbb Z^{V}$ satisfying the conditions in (SBO-JUMP). 
Then, for any $(x, y)$-step $s$, there exists an $(x+s, y)$-step $t$ such that $s+t =p_i$ for some $i \in [\ell]$, and hence $x+s+t = x + p_i \in J$. 
Therefore, if $J$ is a jump system with SBO property, then it is a jump system such that $\sum_{v\in V} x(v)$ has the same parity for any $x \in J$, 
which is called a {\em constant parity jump system}. 
See \cite{MurotaJumpM} for a characterization of constant parity jump systems. 

We now give a few examples of jump systems with SBO property. 

\begin{example}
\label{ex:01}
A matroid $M = (S, \mathcal{B})$ with a ground set $S$ and a base family $\mathcal{B}$ is called \emph{strongly base orderable}
if, for any bases $B_1, B_2 \in \mathcal{B}$, there exists a bijection $f \colon B_1 \setminus B_2 \to B_2 \setminus B_1$ such that 
$(B_1 \setminus X) \cup \{ f(x) \mid x \in X\} \in \mathcal{B}$ for any $X \subseteq B_1 \setminus B_2$ (see e.g.,~\cite[Section 42.6c]{lexbook}). 
By definition, the characteristic vectors of the bases of a strongly base orderable matroid satisfy (SBO-JUMP). 
\end{example}

Note that the characteristic vectors of the bases do not satisfy (SBO-JUMP) if the matroid is not strongly base orderable, 
which implies that the class of jump systems with SBO property is strictly smaller than that of constant parity jump systems. 
By merging some elements in Example~\ref{ex:01}, 
we obtain the following example,  
which was studied for linear matroids in a problem similar to {\sc Jump System Intersection}~\cite{Szabo08}. 

\begin{example}
\label{ex:01'}
Let $M = (S, \mathcal{B})$ be a strongly base orderable matroid and 
let $(S_1, S_2, \dots , S_n)$ be a partition of $S$. 
Then, 
$J = \{x \in \mathbb Z^n \mid  B \in \mathcal{B},\, x(i) = |B \cap S_i|  \mbox{ for } i \in [n]  \}$ satisfies (SBO-JUMP).  
\end{example}

Another example is the set of the degree sequences of subgraphs. 

\begin{example}
\label{ex:02}
Let $G=(V, E)$ be a graph and let $J$ be the set of the degree sequences of subgraphs, i.e., $J = \{d_F \mid F \subseteq E\}$. 
Then, $J$ satisfies (SBO-JUMP). 
To see this, for $x, y \in J$, let $M, N \subseteq E$ be edge sets with $d_M = x$ and $d_N = y$. 
Then, the symmetric difference of $M$ and $N$ can be decomposed into alternating paths $P_1, \dots , P_\ell$ and alternating cycles 
such that $\{d_{N\cap P_i} - d_{M \cap P_i} \mid i \in [\ell] \}$ is a $2$-step decomposition of $y-x$. Note that each $P_i$ is regarded as an edge subset.  
Let $p_i := d_{N\cap P_i} - d_{M \cap P_i}$ for  $i \in [\ell]$. 
For any $I \subseteq [\ell]$,  $x + \sum_{i \in I} p_i$ is the degree sequence of the symmetric difference of 
$M$ and $\bigcup_{i \in [I]} P_i$, and hence it is in $J$. 
\end{example}

Our contribution is to introduce the jump system with SBO property and show that 
(SBO-JUMP) is crucial when we apply the algorithm in \cite{DudyczP17} for {\sc Jump System Intersection}. 
For $\alpha, \beta \in \mathbb Z$ with $\alpha \le \beta$ that have the same parity, 
a set $\{\alpha, \alpha + 2, \dots , \beta-2, \beta\}$ is called 
a {\em parity interval}. 
The main result in this paper is stated as follows.

\begin{theorem}
\label{thm:main}
There is an algorithm for {\sc Jump System Intersection} 
whose running time is polynomial in $\sum_{v \in V} \sum_{\alpha \in B(v)} \log (| \alpha| +1) + \sum_{v \in V} \log  ( |c(v)| +1)$
if the following properties hold: 
\begin{enumerate}
\item[(C1)]
a feasible solution $x_0 \in J \cap B$ is given,  
\item[(C2)]
$J$ satisfies (SBO-JUMP), and  
\item[(C3)]
for any direct product  $B' \subseteq \mathbb Z^V$ of parity intervals, 
there is an oracle for finding a vector $x \in J \cap B'$ maximizing $c^\top x$. 
\end{enumerate}
\end{theorem}

Note that no explicit representation of $J$ is required in this theorem. 
We only need the oracle in Condition (C3). 
Note also that Condition (C3) implies the existence of the membership oracle of $J$.

When $J$ is the set of the degree sequences of subgraphs, we see that $J$ satisfies (C1)--(C3) as follows. 
It was shown by Cornu\'ejols~\cite{CORNUEJOLS1988185} that a feasible solution $x_0 \in J \cap B$ in (C1) can be found  in polynomial time, and 
(C2) holds by Example~\ref{ex:02}. 
The subproblem in (C3) is to find a maximum $(a, b)$-parity factor, 
which can be solved in polynomial time.

Our proof for Theorem~\ref{thm:main} is based on the argument of Dudycz and Paluch~\cite{DudyczP17}. 
While their algorithm is very simple, the correctness proof is quite complicated. 
In particular, an involved case analysis is required to prove a key lemma~\cite[Lemma 2]{DudyczP17}. 
Our technical contribution in this paper is to give a new simpler proof of this lemma
in a slightly different form (Lemma~\ref{prop:localimprove}). 
In our proof, we use several properties that are peculiar to our problem formulation (see Section~\ref{sec:propmin}), 
which is an advantage of introducing the abstract form of the optimal general factor problem. 
As a byproduct of our analysis, 
we show that a scaling technique used in~\cite{DudyczP17} is not required in the algorithm, 
which is another contribution of this paper.

We also introduce a quantitative extension of (SBO-JUMP), 
and extend Theorem~\ref{thm:main} to a valuated variant of {\sc Jump System Intersection}; see Theorem~\ref{thm:mainweighted}.

\subsection{Organization}

The rest of this paper is organized as follows. 
We first give some preliminaries in Section~\ref{sec:preliminary}. 
In Section~\ref{sec:algorithm}, 
we describe our algorithm and prove its correctness by using a key technical lemma (Lemma~\ref{prop:localimprove}). 
A proof of Lemma~\ref{prop:localimprove} is given in Section~\ref{sec:proofkey}, where 
properties shown in Section~\ref{sec:propmin} play important roles to simplify the argument.  
In Section~\ref{sec:valuated}, we extend our results to the valuated case. 
Then, in Section~\ref{sec:weightedGF}, we show that a polynomial-time algorithm for the weighted general factor problem is derived from these results. 
Finally, in Section~\ref{sec:conclusion}, we conclude this paper by giving some remarks.

\section{Preliminaries}
\label{sec:preliminary}

Let $V$ be a finite set. For $v \in V$, let $\chi_v \in \mathbb Z^V$ denote the characteristic vector of $v$, that is, $\chi_v(v)=1$ and $\chi_v(u)=0$ for $u \in V \setminus \{v\}$. 
For each $v \in V$, we are given a non-empty finite set $B(v) \subseteq \mathbb Z$ that has no gap of length more than one, i.e., $B(v)$ is a one-dimensional jump system. 
Throughout this paper, let $B \subseteq \mathbb Z^V$ be the direct product of $B(v)$'s, i.e., $B := \{ x \in \mathbb Z^V \mid x (v) \in B(v) \mbox{ for any $v \in V$} \}$. 
For $x \in \mathbb Z^V$, we denote $\min B \le x \le \max B$ if $\min B(v) \le x(v) \le \max B(v)$ for every $v \in V$.  
For $x \in \mathbb Z^V$, we define $q(x) = |\{ v \in V \mid  x(v) \not\in B(v) \} |$. 
Note that, if $\min B \le x \le \max B$, then  $q(x) := \min_{y \in B} \|x - y\|_1$, because each $B(v)$ has no gap of length greater than one. 
Recall that a parity interval is a subset of $\mathbb Z$ that is of the form 
$\{\alpha, \alpha + 2, \dots , \beta-2, \beta\}$. 
For $v \in V$, we see that $B(v)$ is uniquely partitioned into inclusionwise maximal parity intervals (see Figure~\ref{fig:01}), 
which we call {\em maximal parity intervals of $B(v)$}.
For $\alpha, \beta \in \mathbb Z$ with $\min B(v) \le \alpha \le \beta \le \max B(v)$, we define ${\rm dist}_{B(v)} (\alpha, \beta)$ as the number of 
maximal parity intervals of $B(v)$ intersecting $[\alpha, \beta]$ minus one. 
In other words, ${\rm dist}_{B(v)} (\alpha, \beta)$ is the number of pairs of consecutive integers in $B(v) \cap [\alpha, \beta]$. 
We also define  ${\rm dist}_{B(v)} (\beta, \alpha) := {\rm dist}_{B(v)} (\alpha, \beta)$. 
For $x, y \in \mathbb Z^V$ with $\min B \le x, y \le \max B$, 
we define ${\rm dist}_{B} (x, y) := \sum_{v \in V} {\rm dist}_{B(v)} (x(v), y(v))$; see Figure~\ref{fig:02} for an example. 
Note that ${\rm dist}_{B}$ satisfies the triangle inequality. 
The following lemma is derived from the definitions of ${\rm dist}_{B}$ and $q$. 

\begin{figure}
    \centering
    \includegraphics[width=10cm]{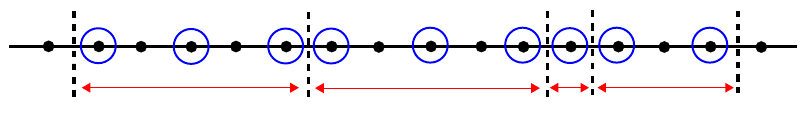}
    \caption{Blue circles are elements in $B(v)$ and red arrows indicate maximal parity intervals.}
    \label{fig:01}
\end{figure}

\begin{figure}
    \centering
    \includegraphics[width=8cm]{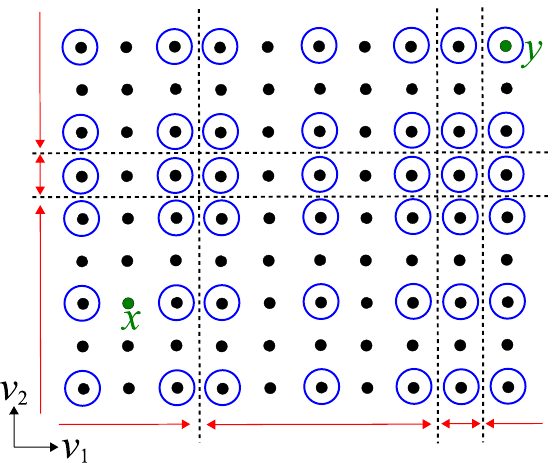}
    \caption{In this two-dimensional example, ${\rm dist}_{B(v_1)} (x(v_1), y(v_1)) = 3$, ${\rm dist}_{B(v_2)} (x(v_2), y(v_2)) = 2$, ${\rm dist}_{B} (x, y) = 5$, $\|x-y\|_1 = 14$, $q(x)=1$, and $q(y)=0$.}
    \label{fig:02}
\end{figure}

\begin{lemma}\label{lem:parity}
For $x, y \in \mathbb Z^V$ with $\min B \le x, y \le \max B$, we have that ${\rm dist}_{B} (x, y) + \|x-y\|_1 +  q(x) + q(y)$ is even. 
\end{lemma}

\begin{proof}
Let $x', y' \in B$ be vectors such that $q(x) = \| x - x'\|_1$ and $q(y) = \| y - y'\|_1$. 
Then, ${\rm dist}_{B} (x, x') = {\rm dist}_{B} (y, y') = 0$, and hence ${\rm dist}_{B} (x, y) = {\rm dist}_{B} (x', y')$. 
Since  $x', y' \in B$ implies that ${\rm dist}_{B(v)} (x'(v), y'(v))$ and  $|x'(v)-y'(v)|$ have the same parity for each $v \in V$, we obtain 
\begin{align*}
&{\rm dist}_{B} (x, y) + \|x-y\|_1 +  q(x) + q(y)  \\
&\quad = {\rm dist}_{B} (x', y') + \|x-y\|_1 +  \|x-x'\|_1 + \|y-y'\|_1 \\
&\quad \equiv {\rm dist}_{B} (x', y') + \|x'-y'\|_1 \\
&\quad = \sum_{v \in V} ({\rm dist}_{B(v)} (x'(v), y'(v)) + |x'(v)-y'(v)|) \\
&\quad \equiv  0 \pmod 2, 
\end{align*}
which completes the proof. 
\end{proof}

\section{Algorithm and Correctness}
\label{sec:algorithm}

Our algorithm for {\sc Jump System Intersection} is basically the same as~\cite{DudyczP17}. 
We first initialize the vector $x:=x_0$, where $x_0$ is as in Condition (C1) in Theorem~\ref{thm:main}. 
In each iteration, 
we compute a vector $x' \in J \cap B$ maximizing $c^\top  x'$ subject to ${\rm dist}_{B} (x, x') \le 2$. 
If $c^\top  x' = c^\top  x$, then the algorithm terminates by returning $x$. 
Otherwise, we replace $x$ with $x'$ and repeat the procedure. 
See Algorithm~\ref{alg:01} for a pseudocode of the algorithm.

\begin{algorithm}
    \KwInput{$J, B$, $c$, and $x_0$.} 
    \KwOutput{$x \in J \cap B$ maximizing $c^\top x$.}
	$x \gets x_0$\;  
	\While{\rm true}{ 
		Find a vector $x' \in J \cap B$ maximizing $c^\top  x'$ subject to ${\rm dist}_{B} (x, x') \le 2$\;
		\If{$c^\top x' = c^\top x$}{\Return{$x$}}
		$x \gets x'$\; 
    	}
    \caption{Algorithm for {\sc Jump System Intersection} }\label{alg:01}
\end{algorithm}

In the correctness proof, we use the following key lemma, whose proof is given in Section~\ref{sec:proofkey}. 
Note again that giving a simpler proof for this lemma is a technical contribution of this paper. 

\begin{lemma}\label{prop:localimprove}
Let $x, y \in B$ be vectors with ${\rm dist}_{B} (x, y)=4$, 
let $\{p_1, \dots , p_\ell\}$ be a $2$-step decomposition of $y-x$, 
and let $w_i \in \mathbb R$ for $i \in [\ell]$. 
Then, there exists a set $I \subseteq [\ell]$ such that 
$z := x + \sum_{i \in I} p_i$  is contained in $B$, 
${\rm dist}_{B} (x, z)= 2$, and 
$\sum_{i \in I} w_i \ge \min \{0, \sum_{i \in [\ell]} w_i\}$. 
\end{lemma}

Let $w \in \mathbb R^{\ell}$ be the vector consisting of $w_i$'s, and 
denote $w(I) := \sum_{i \in I} w_i$ for $I \subseteq [\ell]$.

\begin{remark}
\label{rem:symmetry}
In Lemma~\ref{prop:localimprove}, the roles of $x$ and $y$ are symmetric by changing the signs of $p_i$ and $w_i$, because 
$\bar I := [\ell] \setminus I$ satisfies the following: 
\begin{itemize}
\item
$x + \sum_{i \in I} p_i = y + \sum_{i \in \bar I} (- p_i)$, 
\item
${\rm dist}_{B} (x, z)= 2$ if and only if ${\rm dist}_{B} (y, z)= 2$, and 
\item
$\sum_{i \in I} w_i \ge \min \big\{ 0, \sum_{i \in [\ell]} w_i \big\}$ if and only if 
$\sum_{i \in \bar I} (- w_i) \ge \min \big\{ 0, \sum_{i \in [\ell]} (- w_i) \big\}$. 
\end{itemize}
\end{remark}

We next show the following lemma. 
Note that almost the same result is shown for degree sequences in~\cite[Lemma 1]{DudyczP17}.

\begin{lemma}
\label{lem:reachable}
Let $k$ be a positive integer. Let $x, y \in B$ be vectors with ${\rm dist}_{B} (x, y)=2k$ and 
let $\{p_1, \dots , p_\ell\}$ be a $2$-step decomposition of $y-x$. 
Then, there exist index sets $\emptyset = I_0 \subsetneq I_1 \subsetneq I_2 \subsetneq \dots \subsetneq I_k = [\ell]$ such that 
$z_j := x + \sum_{i \in I_j} p_i$ is contained in $B$ and ${\rm dist}_{B} (z_{j-1}, z_j)= 2$ for $j \in [k]$. 
\end{lemma}

\begin{proof}
It suffices to construct $I_1\subseteq [\ell]$ satisfying the conditions, because $I_2, I_3, \dots, I_{k-1}$ can be constructed in this order in the same way. 

By changing the direction of axes if necessary, we may assume that $x(v) \le y(v)$ for every $v \in V$. 
Then, each $p_i$ is equal to  $\chi_a + \chi_b$ for some $a, b \in V$ (possibly $a=b$). 
For $z \in \mathbb Z^V$, we denote $\phi(z) := ({\rm dist}_{B} (x, z), q(z)) \in \mathbb Z_{\ge 0}^2$. 
In order to construct $I_1$, we start with $I := I_0 = \emptyset$ and add an element one by one to $I$. 
During the procedure, we keep $\phi(z) \in \{ (0, 0), (0, 2), (1, 1), (2, 0)\}$, 
where $z := x+\sum_{i \in I} p_i$. 
Note that $\phi(z) = (0, 0)$ when $I$ is initialized to $I_0$.

If $\phi(z) = (2, 0)$, then $I_1 := I$ clearly satisfies the conditions. 
Otherwise, it holds that $\phi(z) \in \{(0, 0), (0, 2), (1, 1)\}$. 
In this case, we show that there exists an index $i \in [\ell] \setminus I$ such that 
$\phi(z + p_i) \in \{(0, 0), (0, 2), (1, 1), (2, 0)\}$ by the following case analysis. 

\begin{itemize}
\item
Suppose that $\phi(z) = (0, 0)$. Let $i$ be an arbitrary index in $[\ell] \setminus I$. 
Then, $p_i = \chi_a + \chi_b$ for some $a, b \in V$ (possibly $a=b$). 
We see that $\phi(z + \chi_a) \in \{(0, 1), (1, 0)\}$, and hence 
$\phi(z + p_i) = \phi(z +\chi_a + \chi_b) \in \{(0, 0), (0, 2), (1, 1), (2, 0)\}$. 
\item
Suppose that $\phi(z) = (0, 2)$. Then, $z + \chi_a + \chi_b \in B$ for some distinct $a, b \in V$ such that $z(a) < y(a)$ and $z(b) < y(b)$. 
Let $i$ be an index in $[\ell] \setminus I$ such that $p_i = \chi_a + \chi_c$ for some $c \in V$ (possibly $c=a$ or $c=b$). 
Then, we see that $\phi(z + \chi_a) = (0, 1)$, and hence 
$\phi(z + p_i) = \phi(z + \chi_a + \chi_c) \in \{(0, 0), (0, 2), (1, 1)\}$. 
\item
Suppose that $\phi(z) = (1, 1)$. Then, $z + \chi_a \in B$ for some $a \in V$ with $z(a) < y(a)$. 
Let $i$ be an index in $[\ell] \setminus I$ such that $p_i = \chi_a + \chi_b$ for some $b \in V$ (possibly $b=a$). 
Then, we see that $\phi(z + \chi_a) = (1, 0)$, and hence 
$\phi(z + p_i) = \phi(z + \chi_a + \chi_b) \in \{(1, 1), (2, 0)\}$. 
\end{itemize}

If $\phi(z + p_i) = (2, 0)$, then $I_1 := I \cup \{i\}$ satisfies the conditions. 
Otherwise, we replace $I$ with $I \cup \{i\}$ and repeat the procedure. 
Since $[\ell]$ is finite, 
this process terminates by finding a desired index set $I_1$, 
which completes the proof. 
\end{proof}

By using Lemmas~\ref{prop:localimprove} and~\ref{lem:reachable}, 
we can evaluate the improvement of the objective value in each iteration of Algorithm~\ref{alg:01} as follows.

\begin{lemma}
\label{lem:updateratio}
Let $J$ be a jump system with SBO property, 
let $x^* \in J \cap B$ be an optimal solution of {\sc Jump System Intersection}, and   
let $x \in J \cap B$ be a vector with $x \neq x^*$. 
Let $x' \in J \cap B$ be a vector maximizing $c^\top x'$ subject to ${\rm dist}_{B} (x, x') \le 2$. 
Then, $c^\top x' - c^\top x \ge \frac{2}{\|x^* - x \|_1}  (c^\top x^* - c^\top x)$.  
\end{lemma}

\begin{proof}
If ${\rm dist}_{B} (x, x^*) = 2$, then the inequality is obvious. 
Since ${\rm dist}_{B} (x, x^*)$ is even by Lemma~\ref{lem:parity}, suppose that ${\rm dist}_{B} (x, x^*) \ge 4$. 
Since $x, x^* \in J$, there exists a $2$-step decomposition $\{p_1, \dots , p_\ell\}$ of $x^*-x$ that satisfies the conditions in (SBO-JUMP). 
For $i \in [\ell]$, we define $w_i = c^\top p_i - \frac{ c^\top x^* - c^\top x }{\ell} + \varepsilon$, where $\varepsilon$ is a sufficiently small positive number (e.g. $\varepsilon = \frac{1}{(\ell+1)^2}$) that is used to break ties. 
Observe that, for $I, I' \subseteq [\ell]$ with $|I| \neq |I'|$, $w(I) \neq  w(I')$ holds because of $\varepsilon$. 
By Lemma~\ref{lem:reachable}, there exist index sets $\emptyset = I_0 \subsetneq I_1 \subsetneq I_2 \subsetneq \dots \subsetneq I_k = [\ell]$ such that 
$z_j := x + \sum_{i \in I_j} p_i$ is contained in $B$ and ${\rm dist}_{B} (z_{j-1}, z_j)= 2$ for $j \in [k]$.
We choose $I_1, I_2, \dots , I_{k-1}$ so that 
$(w(I_1), w(I_2), \dots , w(I_{k-1}))$ is lexicographically maximum. 
Note that $z_j \in J$ for $j \in [k]$ by (SBO-JUMP).

Let $j \in [k]$ be the minimum index such that $w (I_{j-1})  < w (I_{j})$. 
Note that such $j$ must exist, because $w(I_0) = 0  < \varepsilon \ell = w(I_k)$.
Assume that $j \neq 1$. 
Then, the minimality of $j$ shows that 
$$
w( I_{j-2})  > w( I_{j-1} ) < w( I_j ), 
$$
where we note that $w( I_{j-2} ) \neq w ( I_{j-1})$ as $|I_{j-2}| \neq |I_{j-1}|$. 
By applying Lemma~\ref{prop:localimprove} to a $2$-step decomposition $\{ p_i \mid i \in I_j \setminus I_{j-2} \}$ of $z_j - z_{j-2}$, 
we obtain an index set $I \subseteq I_j \setminus I_{j-2}$ such that 
$z'_{j-1} := z_{j-2} + \sum_{i \in I} p_i$ is contained in $B$,  
${\rm dist}_{B} (z_{j-2}, z'_{j-1})= 2$, and 
$w( I ) \ge \min \{0,  w ( I_j  \setminus I_{j-2} )  \}$. 
Let $I'_{j-1} := I_{j-2} \cup I$. 
By $z'_{j-1} = x + \sum_{i \in I'_{j-1}} p_i$ and (SBO-JUMP), we see that $z'_{j-1} \in J$.  
Furthermore, we obtain 
$$
w( I'_{j-1} ) = w( I_{j-2}) + w( I )  \ge \min \left\{ w( I_{j-2} ) , w( I_{j} ) \right\} > w( I_{j-1} ), 
$$
which contradicts the choice of $I_{j-1}$. 

Therefore, we obtain $j=1$, that is,  $0 = w( I_0 ) < w( I_1 )$. 
Since 
\begin{align*}
0 < w( I_1 ) =   \sum_{i \in I_1} \left( c^\top p_i - \frac{ c^\top x^* - c^\top x }{\ell} + \varepsilon \right) =   c^\top z_1 - c^\top x -  \left( \frac{c^\top x^* - c^\top x}{\ell} - \varepsilon \right)   |I_1| 
\end{align*}
and $\varepsilon$ is sufficiently small, 
we obtain 
\begin{equation*}
c^\top z_1 - c^\top x \ge \frac{(c^\top x^* - c^\top x) |I_1| }{\ell}. 
\end{equation*}
We also see that $c^\top x' \ge c^\top z_1$, because 
$z_1 \in J\cap B$ and ${\rm dist}_{B} (x, z_1) \le 2$. 
By combining these inequalities with $|I_1| \ge 1$ and $\ell = \frac{\|x^* - x \|_1}{2}$, 
we obtain $c^\top x' - c^\top x \ge \frac{2}{\|x^* - x \|_1}  (c^\top x^* - c^\top x)$. 
\end{proof}

This implies that 
the global optimality is guaranteed by the local optimality as follows.

\begin{corollary}
\label{cor:optimality}
In an instance of {\sc Jump System Intersection} with (C2), 
a feasible solution $x \in J \cap B$ maximizes $c^\top x$ if and only if 
$c^\top x \ge c^\top x'$ for any $x' \in J \cap B$ with ${\rm dist}_{B} (x, x') \le 2$. 
\end{corollary}

We are now ready to prove the correctness of Algorithm~\ref{alg:01}. 

\begin{proof}[Proof of Theorem~\ref{thm:main}]
We first show that each iteration of Algorithm~\ref{alg:01} runs in polynomial time. 
For $x, x' \in B$ with ${\rm dist}_B(x, x') \le 2$, we see that
$x(v)$ and $x'(v)$ are contained in the same maximal parity interval of $B(v)$
for any $v \in V$ except at most two elements. 
Thus, for $x \in B$, 
$\{ x' \in B \mid {\rm dist}_B(x, x') \le 2 \}$ can be partitioned into 
$O(n^2)$ sets, each of which is a direct product of parity intervals. 
Therefore, 
we can find a vector $x' \in J \cap B$ maximizing $c^\top  x'$ subject to ${\rm dist}_{B} (x, x') \le 2$ 
by using the oracle in Condition (C3), 
$O(n^2)$ times.

We next evaluate the number of iterations in the algorithm. 
Let ${\rm OPT}$ be the optimal value of the problem and let $B_{\rm size} := \sum_{v \in V} |B(v)|$. 
Since $J$ is a jump system with SBO property by Condition (C2), we can apply Lemma~\ref{lem:updateratio}. 
By this lemma, if $x$ is replaced with $x'$ in line 6 of Algorithm~\ref{alg:01}, 
then 
$$
{\rm OPT} - c^\top x' \le \left( 1 - \frac{2}{\|x^* - x \|_1} \right)  ({\rm OPT} - c^\top x) \le \left( 1 - \frac{1}{B_{\rm size}} \right)  ({\rm OPT} - c^\top x),  
$$
that is, the gap to the optimal value decreases by a factor of at most $1 - \frac{1}{B_{\rm size}}$. 
Therefore, by repeating this procedure $O(B_{\rm size} \log ({\rm OPT} - c^\top x_0))$ times, the algorithm terminates. 
The obtained solution is an optimal solution by Corollary~\ref{cor:optimality}.

This shows that Algorithm~\ref{alg:01} solves {\sc Jump System Intersection} in polynomial time. 
\end{proof}

\section{Proof of Lemma~\ref{prop:localimprove}}
\label{sec:proofkey}

In this section, we give a proof of Lemma~\ref{prop:localimprove}. 
In this lemma, a tuple $(x, y, (p_i)_{i \in [\ell]}, w)$ is called an {\it instance} and 
a set $I$ satisfying the conditions is called a {\it solution}.


\subsection{Minimal Counterexample}
\label{sec:propmin}

To derive a contradiction, assume that Lemma~\ref{prop:localimprove} does not hold. 
Suppose that $(x, y, (p_i)_{i \in [\ell]}, w)$ is a counterexample that minimizes $\|y-x\|_1$. 
Among such counterexamples, we choose one that minimizes $|\{ (p_i, w_i) \mid i \in [\ell]\}|$, that is, 
we minimize the number of different $(p_i, w_i)$ pairs. 
Such $(x, y, (p_i)_{i \in [\ell]}, w)$ is called a {\it minimal counterexample}. 
Define $U \subseteq V$ as $U := \{v \in V \mid {\rm dist}_{B(v)} (x(v), y(v)) \ge 1 \}$. 
By changing the direction of axes if necessary, we may assume that $x(v) \le y(v)$ for every $v \in V$. 
Then, each $p_i$ is equal to  $\chi_a + \chi_b$ for some $a, b \in V$ (possibly $a=b$). 
We show some properties of the minimal counterexample. 
Our argument becomes simpler with the aid of these properties.

\begin{claim}\label{clm:03}
For any $i \in [\ell]$,  $p_i = \chi_a + \chi_b$ for some $a, b \in U$ (possibly $a = b$). 
Consequently, $x(v)=y(v)$ for all $v\in V\setminus U$. 
\end{claim}

\begin{proof}
Assume to the contrary that there exists $i \in [\ell]$ such that  $p_i = \chi_a + \chi_c$ for some $a \in V$ and for some $c \in V \setminus U$. 

Suppose that $a = c$, i.e., $p_i = 2 \chi_c$. 
We consider a new instance by removing $p_i$ and replacing $y$ with $y - 2 \chi_c \in B$. 
By the minimality of the counterexample, 
the obtained instance has a solution $I \subseteq [\ell] \setminus \{i\}$, which implies that 
$w( I ) \ge 0$ or $ w( I ) \ge w (  [\ell] \setminus \{i\} )$. 
Then, $I' := I$ is a solution of the original instance in the former case and 
$I' := I \cup \{i\}$ is a solution of the original instance in the latter case, 
which is a contradiction. 

Suppose next that $a \not= c$. 
Since ${\rm dist}_{B(c)} (x(c), y(c))=0$ and $x(c), y(c) \in B(c)$, we see that $x(c)$ and $y(c)$ have the same parity. 
Thus, there exists $i' \in [\ell] \setminus \{i\}$ such that $p_{i'}=\chi_b + \chi_c$ for some $b \in V \setminus \{c\}$. 
We merge $p_i$ and $p_{i'}$ as follows: 
replace $p_i$ and $p_{i'}$ with a new vector $p_{i''} := \chi_a + \chi_b$ whose weight is $w_i + w_{i'}$, and replace $y$ with $y - 2 \chi_c \in B$. 
By the minimality of the counterexample, 
the obtained instance has a solution $I \subseteq ([\ell] \setminus \{i, i'\}) \cup \{i''\}$. 
Then, we see that the set 
$$
I' := \begin{cases}
(I  \setminus \{i''\}) \cup \{i, i'\}  & \mbox{if $i'' \in I$,} \\
I   & \mbox{otherwise}
\end{cases}
$$
is a solution of the original instance, which is a contradiction. 
\end{proof}

\begin{claim}\label{clm:04}
For any $i \in [\ell]$ and for any $a \in U$ with ${\rm dist}_{B(a)} (x(a), y(a))=1$,  $p_i \not= 2 \chi_a$. 
\end{claim}

\begin{proof}
Assume to the contrary that $p_i = 2 \chi_a$ for some $a \in U$ with ${\rm dist}_{B(a)} (x(a), y(a))=1$. 
Suppose first that $w_i \ge 0$. 
We consider the following two cases separately. 

\medskip
\noindent
\textbf{Case 1: }
Suppose that $x(a) + 1 \not\in B(a)$, which implies that $x(a) + 2 \in B(a)$. 
We construct a new instance by removing $p_i$ and replacing $x$ with $x + 2 \chi_a \in B$. 
By the minimality of the counterexample, 
the obtained new instance has a solution $I \subseteq [\ell] \setminus \{i\}$. 
Then, we see that $I \cup \{i\}$ is a solution of the original instance, because 
$$
w( I \cup \{i\} ) - \min \{ 0, w ( [\ell] ) \} 
\ge w(I)- \min \{0, w(  [\ell] \setminus \{i\} ) \} 
\ge 0,  
$$
where the first inequality follows from  $w_i  \ge 0$ and the second inequality follows from the fact that $I$ is a solution of the new instance. 
This is a contradiction. 

\medskip
\noindent
\textbf{Case 2: }
Suppose that $x(a) + 1 \in B(a)$, which means that $x(a)$ and $x(a)+1$ are contained in different maximal parity intervals of $B(a)$. 
Since  ${\rm dist}_{B(a)} (x(a), y(a))=1$ and $x(a), y(a) \in B(a)$, it holds that $x(a) \not\equiv y(a) \pmod 2$.
This implies that there exists $i' \in [\ell] \setminus \{i\}$ such that $p_{i'} = \chi_a + \chi_b$ for some $b \in V \setminus \{a\}$. 
We construct a new instance as follows: 
remove $p_i$, replace $w_{i'}$ with $w_{i} + w_{i'}$, and replace $y$ with $y - 2 \chi_a$. 
Note that $y - 2 \chi_a \in B$, because $y(a)$ and $x(a)+1$ are contained in the same maximal parity interval of $B(a)$ and $y(a) \ge x(a) + 3$. 
By the minimality of the counterexample, 
the obtained instance has a solution $I \subseteq [\ell] \setminus \{i\}$. 
Then, we see that the set 
$$
I' := \begin{cases}
I  \cup \{i\}  & \mbox{if $i' \in I$,} \\
I   & \mbox{otherwise}
\end{cases}
$$
is a solution of the original instance, which is a contradiction.

\medskip

These two cases complete the proof for the case when $w_i \ge 0$. 
When $w_i \le 0$, we can apply the same argument as above by changing the roles of $x$ and $y$ (see Remark~\ref{rem:symmetry}). 
Therefore, the claim holds. 
\end{proof}

\begin{claim}\label{clm:01}
For any $i, j \in [\ell]$ with $p_i = p_j$, it holds that $w_i = w_j$.  
\end{claim}

\begin{proof}
Let $(x, y, (p_i)_{i \in [\ell]}, w)$ be a minimal counterexample of Lemma~\ref{prop:localimprove}, and 
assume that $p_i = p_j$ does not imply $w_i = w_j$. 
Let $I^*  \subseteq [\ell]$ be a maximal index set such that 
$p_i = p_j$ for any $i, j \in I^*$ and  $w_i \neq w_j$ for some $i, j \in I^*$. 
We denote $I^* = \{i_1, i_2, \dots , i_t\}$, where $w_{i_1} \ge w_{i_2} \ge \dots \ge w_{i_t}$.  
Let $w^* := \frac{1}{t} w( I^* )$. Define $w'_{i} := w^*$ for $i \in I^*$ and $w'_i := w_i$ for $i \in [\ell] \setminus I^*$. 
We note that $w'( [\ell] ) = w( [\ell] )$. 
If there exists a solution $I' \subseteq [\ell]$ for a new instance $(x, y, (p_i)_{i \in [\ell]}, w')$,  
then $I := (I' \setminus I^*) \cup \{i_1, i_2, \dots , i_{|I' \cap I^*|}\}$ is a solution for the original instance $(x, y, (p_i)_{i \in [\ell]}, w)$,  
because $w_{i_1} + w_{i_2} + \dots + w_{i_{|I' \cap I^*|}} \ge |I' \cap I^*| \cdot w^*  = w' ( I' \cap I^* )$ 
implies that $w (I) \ge w( I' )$. 
This shows that instance $(x, y, (p_i)_{i \in [\ell]}, w')$ has no solution, and hence it is a counterexample. 
Since $|\{ (p_i, w'_i) \mid i \in [\ell]\}| < |\{ (p_i, w_i) \mid i \in [\ell]\}|$, this contradics the minimality of $(x, y, (p_i)_{i \in [\ell]}, w)$. 
\end{proof}

Let $I^+  := \{ i \in [\ell] \mid w_i > 0\}$ and $z^+ := x + \sum_{i \in I^+} p_i$. 
By this claim, we observe the following. 

\begin{observation}
\label{obs:01}
For any $i \in I^+$ and for any $j \in [\ell] \setminus I^+$, 
it holds that $p_i \neq p_j$. 
\end{observation}

Since ${\rm dist}_{B} (x, z^+) + {\rm dist}_{B} (y, z^+) = 4$, by changing the roles of $x$ and $y$ if necessary (see Remark~\ref{rem:symmetry}), 
we may assume that ${\rm dist}_{B} (x, z^+) \le 2$.%
\footnote{If we change the roles of $x$ and $y$, then $I^-  := \{ i \in [\ell] \mid w_i < 0\}$ and  $z^- := y - \sum_{i \in I^-} p_i$ play the roles of  $I^+$ and $z^+$, respectively. 
We see that if ${\rm dist}_{B} (x, z^+) \ge 3$, then ${\rm dist}_{B} (y, z^-) \le {\rm dist}_{B} (y, z^+) = 4 - {\rm dist}_{B} (x, z^+) \le 1$.}
Furthermore, since $x(v) = y(v) = z^+(v)$ for $v \in V \setminus U$ by Claim~\ref{clm:03}, it holds that $q(z^+) \le |U| \le {\rm dist}_{B} (x, y) = 4$. 
Since $\|x - z^+\|_1$ is even, Lemma~\ref{lem:parity} shows that ${\rm dist}_{B} (x, z^+) + q(z^+)$ is even. 
Overall, the pair $\phi(z^+) := ({\rm dist}_{B} (x, z^+), q(z^+))$ is one of the following: 
$(0, 0), (0, 2), (0, 4), (1, 1), (1, 3), (2, 0), (2, 2)$, and $(2, 4)$. 
Note that  we denote $\phi(z) := ({\rm dist}_{B} (x, z), q(z)) \in \mathbb Z_{\ge 0}^2$ for $z \in \mathbb Z^V$.

In what follows,  
we derive a contradiction for the cases when $|U|=4$, $|U|=3$, and $|U| \le 2$ in Sections~\ref{sec:U=4}, \ref{sec:U=3}, and \ref{sec:U=2}, respectively. 
In the case analysis, we use the following claim, 
which is obtained by the same argument as Lemma~\ref{lem:reachable}.

\begin{claim}\label{clm:02}
Let $I_0 \subseteq [\ell]$ be an index set such that $z_0 := x+\sum_{i \in I_0} p_i$ satisfies  
$\phi(z_0) \in \{ (0, 0), (0, 2), (1, 1), (2, 0)\}$. 
Then, there exists an index set $I \subseteq [\ell]$ with $I_0 \subseteq I$ such that 
$z := x + \sum_{i \in I} p_i$ is contained in $B$ and 
${\rm dist}_{B} (x, z)= 2$, i.e., $\phi(z) = (2, 0)$. 
\end{claim}

\subsection{When $|U| = 4$}
\label{sec:U=4}

Suppose that $|U|=4$, that is, $U = \{v_1, v_2, v_3, v_4\}$ and ${\rm dist}_{B(v_j)} (x(v_j), y(v_j))=1$ for $j\in \{1, 2, 3, 4\}$. 
By Claims~\ref{clm:03} and~\ref{clm:04}, 
for any $i \in [\ell]$,  $p_i = \chi_a + \chi_b$ for some distinct $a, b \in U$. 
As discussed in Section~\ref{sec:propmin}, 
the pair $\phi(z^+) = ({\rm dist}_{B} (x, z^+), q(z^+))$ is one of the following: 
$(0, 0), (0, 2), (0, 4), (1, 1), (1, 3),$  $(2, 0), (2, 2)$, and $(2, 4)$. 
We derive a contradiction by considering each case separately.

\medskip
\noindent
\textbf{Case 1:}  $\phi(z^+) = (0, 0), (0, 2), (1, 1)$, or $(2, 0)$.

By Claim~\ref{clm:02}, 
there exists an index set $I \subseteq [\ell]$ with $I^+ \subseteq I$ such that 
$z := x + \sum_{i \in I} p_i$ is contained in $B$ and 
${\rm dist}_{B} (x, z)= 2$. 
Since $w_i \le 0$ for each $i \in [\ell] \setminus I$, 
we obtain $w( I ) \ge w( [\ell] )$, and hence 
$I$ is a solution of Lemma~\ref{prop:localimprove}. 
This is a contradiction.

\medskip
\noindent
\textbf{Case 2:}  $\phi(z^+) = (0, 4)$.

Let $i \in [\ell] \setminus I^+$. Since $p_i = \chi_a + \chi_b$ for some distinct $a, b \in U$, we obtain
$\phi(z^+ + p_i) = (0, 2)$. 
By Claim~\ref{clm:02}, 
there exists an index set $I \subseteq [\ell]$ with $I^+ \cup \{i\} \subseteq I$ such that 
$z := x + \sum_{j \in I} p_j$ is contained in $B$ and 
${\rm dist}_{B} (x, z)= 2$. 
We see that such $I$ is a solution of Lemma~\ref{prop:localimprove} in the same way as Case 1, 
which is a contradiction. 

\medskip
\noindent
\textbf{Case 3:}  $\phi(z^+) = (1, 3)$.

Without loss of generality, we may assume that $z^+(v_j) \not\in B(v_j)$ for $j \in \{1, 2, 3\}$ and $z^+(v_4) \in B(v_4)$. 
Note that $z^+(v_j) \not\in B(v_j)$ implies $z^+(v_j) \neq x(v_j)$ and $z^+(v_j) \neq y(v_j)$. 
If there exists an index $i \in [\ell] \setminus I^+$ such that 
$p_i \in \{\chi_{v_1} + \chi_{v_2}, \chi_{v_2} + \chi_{v_3}, \chi_{v_3} + \chi_{v_1}\}$, then 
$\phi(z^+ + p_i) = (1, 1)$, and hence 
we can derive a contradiction by applying Claim~\ref{clm:02} in the same say as Case 1.  

Otherwise,  since $z^+(v_j) \not= y(v_j)$ for $j \in \{1, 2, 3\}$, 
there exist $i_1, i_2, i_3 \in [\ell] \setminus I^+$ such that 
$p_{i_1} = \chi_{v_1} + \chi_{v_4}$, 
$p_{i_2} = \chi_{v_2} + \chi_{v_4}$, and 
$p_{i_3} = \chi_{v_3} + \chi_{v_4}$. 
We consider the following two cases separately. 

\begin{itemize}
\item
Suppose that $z^+(v_4) + 1 \not\in B(v_4)$, which implies that $z^+(v_4) + 2 \in B(v_4)$. 
In this case, we obtain $\phi(z^+ + p_{i_1} + p_{i_2}) = (1, 1)$, 
and hence we can derive a contradiction by applying Claim~\ref{clm:02} to $I^+ \cup \{i_1, i_2\}$ in the same way as Case 1. 

\item
Suppose that $z^+(v_4) + 1 \in B(v_4)$, which means that $z^+(v_4)$ and $z^+(v_4)+1$ are contained in different maximal parity intervals of $B(v_4)$. 
Since ${\rm dist}_{B(v_4)} (x(v_4), y(v_4))=1$, this shows that $z^+(v_4)+3 \in B(v_4)$, and hence $\phi(z^+ + p_{i_1} + p_{i_2} + p_{i_3}) = (2, 0)$. 
Therefore, $I := I^+ \cup \{i_1, i_2, i_3\}$ is a solution, because $w(I) \ge w([\ell])$. 
This is a contradiction. 
\end{itemize}

\medskip
\noindent
\textbf{Case 4:}   $\phi(z^+) = (2, 2)$.

Without loss of generality, we assume that 
$z^+(v_j) \not\in B(v_j)$ for $j \in \{1, 2\}$ and $z^+(v_j) \in B(v_j)$ for $j \in \{3, 4\}$. 
If there exists $i \in I^+$ such that $p_{i} = \chi_{v_1} + \chi_{v_2}$, then 
$I := I^+ \setminus \{i\}$ is a solution, because $w( I ) \ge 0$. 
If there exists $i \in [\ell] \setminus I^+$ such that $p_{i} = \chi_{v_1} + \chi_{v_2}$, then 
$I := I^+ \cup \{i\}$ is a solution, because $w ( I )  \ge w( [\ell] )$. 
Therefore, it suffices to consider the case when  $p_i \neq \chi_{v_1} + \chi_{v_2}$ for any $i \in [\ell]$. 

Since $z^+(v_1) \not= x(v_1)$, there exists an index $i_1 \in I^+$ such that $p_{i_1} = \chi_{v_1} + \chi_{v}$ for some $v \in \{v_3, v_4\}$. 
Similarly, since $z^+(v_1) \not= y(v_1)$, there exists an index $i_2 \in [\ell] \setminus I^+$ such that $p_{i_2} = \chi_{v_1} + \chi_{v}$ for some $v \in \{v_3, v_4\}$. 
By Observation~\ref{obs:01}, 
by changing the roles of $v_3$ and $v_4$ if necessary, we may assume that $p_{i_1} = \chi_{v_1} + \chi_{v_3}$ and $p_{i_2} = \chi_{v_1} + \chi_{v_4}$. 
By applying the same argument to $v_2$ instead of $v_1$, 
there exist indices $i_3 \in I^+$ and $i_4 \in [\ell] \setminus I^+$ such that 
$p_{i_3} = \chi_{v_2} + \chi_{v}$ for some $v \in \{v_3, v_4\}$ and $p_{i_4} = \chi_{v_2} + \chi_{v}$ for some $v \in \{v_3, v_4\}$.  
By Observation~\ref{obs:01} again, 
we have either (i) $p_{i_3} = \chi_{v_2} + \chi_{v_4}$ and $p_{i_4} = \chi_{v_2} + \chi_{v_3}$, or 
(ii) $p_{i_3} = \chi_{v_2} + \chi_{v_3}$ and $p_{i_4} = \chi_{v_2} + \chi_{v_4}$.  
We consider each case separately.

\begin{itemize}
\item
Suppose that  $p_{i_3} = \chi_{v_2} + \chi_{v_4}$ and $p_{i_4} = \chi_{v_2} + \chi_{v_3}$; see Figure~\ref{fig:04} (left).\footnote{In figures, a blue edge $(u, v)$ corresponds to an element $i \in [\ell] \setminus I^+$ with $p_i = \chi_u + \chi_v$, 
a red dashed edge $(u, v)$ corresponds to an element $i \in I^+$ with $p_i = \chi_u + \chi_v$, and 
a vertex $v \in V$ in a rectangle satisfies that $z^+(v) \not\in B(v)$.} 
Define  $I_1:= (I^+ \setminus \{i_1\}) \cup \{i_4\} $ and  $I_2:= (I^+ \setminus \{i_2\} ) \cup \{i_3\}$. 
Since $z^+(v_1) \pm 1 \in B(v_1)$ and $z^+(v_2) \pm 1 \in B(v_2)$, 
$z_h := x + \sum_{i \in I_h} p_i$ is contained in $B$ and ${\rm dist}_{B} (x, z_h) = 2$ for $h \in \{1, 2\}$.  
Furthermore, we obtain 
\begin{align}
\max  \big\{ w( I_1 ), w( I_2 )  \big\} 
& \ge \frac{1}{2} \big( w (  I_1 ) + w( I_2 )  \big) \notag \\
& = \frac{1}{2} \Big( w(  [\ell] )  + w( I_1 \cap I_2 )  - w \big( [\ell] \setminus (I_1 \cup I_2) \big) \Big) \label{eq:01} \\
& \ge \frac{1}{2}  w( [\ell] )  \ge \min \big\{0,  w( [\ell] ) \big\}, \notag  
\end{align}
where we use $I_1 \cap I_2 \subseteq I^+ \subseteq I_1 \cup I_2$ in the second inequality. 
This shows that at least one of $I_1$ and $I_2$ is a solution, which is a contradiction. 

\item
Suppose that  $p_{i_3} = \chi_{v_2} + \chi_{v_3}$ and $p_{i_4} = \chi_{v_2} + \chi_{v_4}$; see Figure~\ref{fig:04} (right). 
If $z^+(v_3) - 1 \not\in B(v_3)$, then $z^+(v_3) - 2 \in B(v_3)$, and hence 
$I := I^+ \setminus \{i_1, i_3\}$ is a solution, because $w( I ) \ge 0$. 
If $z^+(v_4) + 1 \not\in B(v_4)$, then $z^+(v_4) + 2 \in B(v_4)$, and hence 
$I := I^+ \cup \{i_2, i_4\}$ is a solution, because $w( I )  \ge w( [\ell] )$. 
Therefore, we have that $z^+(v_3) - 1 \in B(v_3)$ and $z^+(v_4) + 1 \in B(v_4)$, that is, 
$z^+(v_3) - 1$ and $z^+(v_3)$ are contained in different maximal parity intervals of $B(v_3)$, and 
$z^+(v_4)$ and $z^+(v_4) + 1$ are contained in different maximal parity intervals of $B(v_4)$ (Figure~\ref{fig:05}). 
Define $I_1:= (I^+ \setminus \{i_1\}) \cup \{i_4\} $ and  $I_2:= (I^+ \setminus \{i_3\} ) \cup \{i_2\}$. 
Since  $z_h := x + \sum_{i \in I_h} p_i$ is contained in $B$ and ${\rm dist}_{B} (x, z_h) = 2$ for $h \in \{1, 2\}$, 
by the same calculation as (\ref{eq:01}), we see that at least one of $I_1$ and $I_2$ is a solution. 
This is a contradiction. 
\end{itemize}

\begin{figure}
    \centering
    \includegraphics[width=4cm]{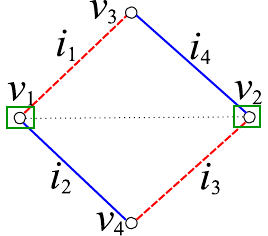}
\qquad \qquad
    \includegraphics[width=4cm]{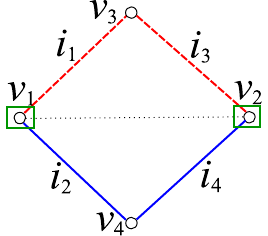}
    \caption{Possible situations in Case 4.}
    \label{fig:04}
\end{figure}

\begin{figure}
    \centering
    \includegraphics[width=5.5cm]{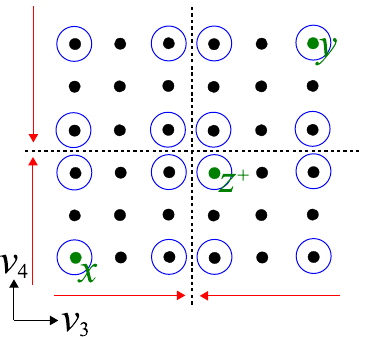}
    \caption{Projection of $z^+$ onto the $v_3$-$v_4$ plane.}
    \label{fig:05}
\end{figure}

\medskip
\noindent
\textbf{Case 5:}  $\phi(z^+) = (2, 4)$.

In this case, $z^+(v) \not\in B(v)$ for $v \in U$. 
For $v \in U$, since $z^+(v) \not= x(v)$, 
there exists $i \in I^+$ such that $p_i = \chi_{v} + \chi_{u}$ for some $u \in U \setminus \{v\}$.
If there exist $i_1, i_2 \in I^+$ with $p_{i_1} + p_{i_2} = \chi_{v_1} + \chi_{v_2} + \chi_{v_3} + \chi_{v_4}$, then 
$I := I^+ \setminus \{i_1, i_2\}$ is a solution, because $w ( I ) \ge 0$, which is a contradiction. 
Otherwise, consider a graph $G=(U, E)$ such that the vertex set is $U$ and $(v, v') \in E$ if there exists $i \in I^+$ with $p_i = \chi_{v} + \chi_{v'}$ (i.e., the graph consisting of the red edges). 
Since the conditions mean that $G$ has no isolated vertex and has no perfect matching, $G$ must be isomorphic to a claw  i.e., a $K_{1,3}$. 
Thus, by changing the roles of $v_1, v_2, v_3$, and $v_4$ if necessary, 
there exist $i_1, i_2, i_3 \in I^+$ such that
$p_{i_1} = \chi_{v_1} + \chi_{v_4}$, 
$p_{i_2} = \chi_{v_2} + \chi_{v_4}$, and 
$p_{i_3} = \chi_{v_3} + \chi_{v_4}$ (Figure~\ref{fig:06}).
This together with Observation~\ref{obs:01} shows that 
there exists no index $i \in [\ell] \setminus I^+$ such that $p_i = \chi_{v_4} + \chi_{v}$ with $v \in \{v_1, v_2, v_3\}$, 
which implies $z^+(v_4) = y(v_4)$. 
This is a contradiction, because $z^+(v_4) \not\in B(v_4)$ and $y(v_4) \in B(v_4)$. 

\begin{figure}[t]
    \centering
    \includegraphics[width=3.8cm]{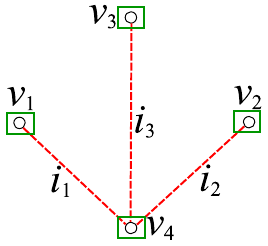}
    \caption{Red edges in Case 5.}
    \label{fig:06}
\end{figure}

\subsection{When $|U|=3$}
\label{sec:U=3}

Suppose that $|U|=3$. Let $U = \{v_1, v_2, v_3\}$ such that ${\rm dist}_{B(v_1)} (x(v_1), y(v_1)) = {\rm dist}_{B(v_2)} (x(v_2), y(v_2))=1$ and ${\rm dist}_{B(v_3)} (x(v_3), y(v_3))=2$. 
By Claims~\ref{clm:03} and~\ref{clm:04}, 
for any $i \in [\ell]$,  either $p_i = \chi_a + \chi_b$ for some distinct $a, b \in U$ or $p_i = 2 \chi_{v_3}$. 
As discussed in Section~\ref{sec:propmin}, since $q(z^+) \le |U| = 3$,  
the pair $\phi(z^+) = ({\rm dist}_{B} (x, z^+), q(z^+))$ is one of the following: 
$(0, 0), (0, 2), (1, 1), (1, 3), (2, 0)$, and $(2, 2)$. 
We derive a contradiction by considering each case separately.

\medskip
\noindent
\textbf{Case 1:} $\phi(z^+) =  (0, 0), (0, 2), (1, 1)$, or $(2, 0)$. 

In this case, we obtain a solution of Lemma~\ref{prop:localimprove} by applying Claim~\ref{clm:02} 
in the same way as Case 1 in Section~\ref{sec:U=4}, which is a contradiction.

\medskip
\noindent
\textbf{Case 2:} $\phi(z^+) =  (1, 3)$. 

In this case, $z^+(v) \not\in B(v)$ for $v \in U$. 
Since $z^+(v_1) \not= y(v_1)$, 
there exists $i \in [\ell] \setminus I^+$ such that $p_i = \chi_{v_1} + \chi_{u}$ for some $u \in \{v_2, v_3\}$.
Then, since $\phi(z^+ + p_i) = (1, 1)$, 
we obtain a solution by applying Claim~\ref{clm:02} to $I^+ \cup \{i\}$ in the same way as Case 1 in Section~\ref{sec:U=4}, 
which is a contradiction.

\medskip
\noindent
\textbf{Case 3:} $\phi(z^+) =  (2, 2)$. 

Since $q(z^+) = 2$ and $|U| = 3$, at least one of $z^+(v_1) \not\in B(v_1)$ and $z^+(v_2) \not\in B(v_2)$ holds. 
By changing the roles of $v_1$ and $v_2$ if necessary, we may assume that $z^+(v_1) \not\in B(v_1)$.  
Let $v^* \in \{v_2, v_3\}$ be the other element such that $z^+(v^*) \not\in B(v^*)$. 
Since $z^+(v_1) \not= x(v_1)$, 
there exists $i_1 \in I^+$ such that $p_{i_1} = \chi_{v_1} + \chi_{u}$ for some $u \in \{v_2, v_3\}$.
Similarly, since $z^+(v_1) \not= y(v_1)$, 
there exists $i_2 \in [\ell] \setminus I^+$ such that $p_{i_2} = \chi_{v_1} + \chi_{u}$ for some $u \in \{v_2, v_3\}$.
By Observation~\ref{obs:01}, 
either $p_{i_1} = \chi_{v_1} + \chi_{v^*}$ or $p_{i_2} = \chi_{v_1} + \chi_{v^*}$ holds (Figure~\ref{fig:07}). 
If  $p_{i_1} = \chi_{v_1} + \chi_{v^*}$, then $I := I^+ \setminus \{i_1\}$ is a solution, because $w ( I ) \ge 0$, which is a contradiction; see Figure~\ref{fig:07} (left two). 
If  $p_{i_2} = \chi_{v_1} + \chi_{v^*}$, then $I := I^+ \cup \{i_2\}$ is a solution, because $w( I  ) \ge w( [\ell] )$, which is a contradiction; see Figure~\ref{fig:07} (right two).

\begin{figure}
    \centering
    \includegraphics[width=13cm]{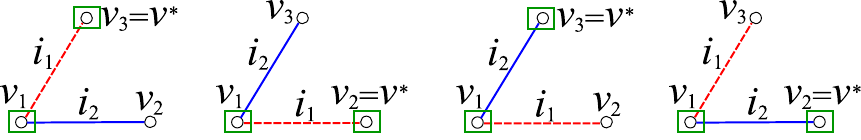}
    \caption{Possible situations in Case 3.}
    \label{fig:07}
\end{figure}


\subsection{When $|U| \le 2$}
\label{sec:U=2}

Suppose that $|U|\le 2$. 
As discussed in Section~\ref{sec:propmin}, since $q(z^+) \le |U| \le 2$,   
the pair $\phi(z^+) = ({\rm dist}_{B} (x, z^+), q(z^+))$ is one of the following: 
$(0, 0), (0, 2), (1, 1), (2, 0)$, and $(2, 2)$. 

If $\phi(z^+) =  (0, 0), (0, 2), (1, 1)$, or $(2, 0)$, then 
we obtain a solution of Lemma~\ref{prop:localimprove} 
in the same way as Case 1 in Section~\ref{sec:U=4}, which is a contradiction. 

Thus, the remaining case is when $\phi(z^+) = (2, 2)$, which implies that $|U|=2$ and $z^+(v) \not\in B(v)$ for $v \in U$. 
Let $U=\{v_1, v_2\}$. 
By Claim~\ref{clm:03}, for any $i \in [\ell]$, it holds that $p_i = \chi_a + \chi_b$ for some $a, b \in U$. 
Since $z^+(v_1) \not= x(v_1)$, 
there exists $i_1 \in I^+$ such that $p_{i_1} = \chi_{v_1} + \chi_{u}$ for some $u \in \{v_1, v_2\}$.
Similarly, since $z^+(v_1) \not= y(v_1)$, 
there exists $i_2 \in [\ell] \setminus I^+$ such that $p_{i_2} = \chi_{v_1} + \chi_{u}$ for some $u \in \{v_1, v_2\}$.
By Observation~\ref{obs:01}, 
either $p_{i_1} = \chi_{v_1} + \chi_{v_2}$ or $p_{i_2} = \chi_{v_1} + \chi_{v_2}$ holds. 
If  $p_{i_1} = \chi_{v_1} + \chi_{v_2}$, then $I := I^+ \setminus \{i_1\}$ is a solution, because $w(I) \ge 0$, which is a contradiction. 
If  $p_{i_2} = \chi_{v_1} + \chi_{v_2}$, then $I := I^+ \cup \{i_2\}$ is a solution, because $w(I)  \ge w( [\ell] )$, which is a contradiction.

\section{Extension to Valuated Problem}
\label{sec:valuated}

In this section, we consider the following valuated version of {\sc Jump System Intersection}.

\begin{problem}
\problemtitle{\sc Valuated Jump System Intersection}
\probleminput{A function $f \colon J \to \mathbb{Z}$ on a jump system $J \subseteq \mathbb Z^V$ and a finite one-dimensional jump system $B(v) \subseteq \mathbb{Z}$ for each $v \in V$.}
\problemoutput{Find a vector $x \in B$ maximizing $f(x)$, where $B \subseteq \mathbb Z^V$ is the direct product of $B(v)$'s. }
\end{problem}

Note that $f$ and $J$ may be given in an implicit way, e.g., by an oracle. 
To simplify the notation, we extend the domain of $f$ to $\mathbb{Z}^V$ 
by setting $f(x) = - \infty$ for $x \in \mathbb{Z}^V \setminus J$. 
The following property is a quantitative extension of  (SBO-JUMP). 

\begin{description}
\item[(SBO-M-JUMP)]
For any $x, y \in J$,  
there exist a $2$-step decomposition  $\{p_1, \dots , p_\ell \}$ of $y-x$ 
and $g_i \in \mathbb{R}$ for $i \in [\ell]$
such that 
$f(x + \sum_{i \in I} p_i) \ge f(x) + \sum_{i \in I} g_i$ for any $I \subseteq [\ell]$ and 
$f(y) = f(x) + \sum_{i \in [\ell]} g_i$.
\end{description}

Note that we use ``M'' in the name of the exchange axiom, because it defines a subclass of 
{\em M-concave functions on constant parity jump systems}~\cite{MurotaJumpM}; see Remark~\ref{rem:mjump} below. 
We can see that if $f$ satisfies (SBO-M-JUMP), then its effective domain $J := \{ x \in \mathbb{Z}^V \mid f(x) > - \infty \}$ satisfies (SBO-JUMP). 
By using (SBO-M-JUMP), we  generalize Theorem~\ref{thm:main} as follows.

\begin{theorem}
\label{thm:mainweighted}
There is an algorithm for {\sc Valuated Jump System Intersection} 
whose running time is polynomial in $\sum_{v \in V} \sum_{\alpha \in B(v)} \log ( | \alpha | + 1) + \max_{x \in J} \log ( |f(x)| + 1)$
if the following properties hold: 
\begin{enumerate}
\item[(C1')]
a vector $x_0 \in J \cap B$ is given,  
\item[(C2')]
$f$ satisfies (SBO-M-JUMP), and  
\item[(C3')]
for any direct product  $B' \subseteq \mathbb Z^V$ of parity intervals, 
there is an oracle for finding a vector $x \in J \cap B'$ maximizing $f(x)$. 
\end{enumerate}
\end{theorem}

Our algorithm for {\sc Valuated Jump System Intersection} is almost the same as Algorithm~\ref{alg:01}, 
where we just replace the objective function $c^\top x$ with $f(x)$; see Algorithm~\ref{alg:02}. 
In order to prove Theorem~\ref{thm:mainweighted}, we extend Lemma~\ref{lem:updateratio} to the valuated case as follows.

\begin{algorithm}
    \KwInput{$f \colon J \to \mathbb Z^V$, $B$, and $x_0$.} 
    \KwOutput{$x \in J \cap B$ maximizing $f(x)$.}
	$x \gets x_0$\;  
	\While{\rm true}{ 
		Find a vector $x' \in J \cap B$ maximizing $f(x')$ subject to ${\rm dist}_{B} (x, x') \le 2$\;
		\If{$f(x') = f(x)$}{\Return{$x$}}
		$x \gets x'$\; 
    	}
    \caption{Algorithm for {\sc Valuated Jump System Intersection} }\label{alg:02}
\end{algorithm}

\begin{lemma}
\label{lem:updateratioweighted}
Let $f \colon J \to \mathbb Z$ be a function satisfying (SBO-M-JUMP), 
let $x^* \in J \cap B$ be an optimal solution of {\sc Valuated Jump System Intersection}, and  
let $x \in J \cap B$ be a vector with $x \neq x^*$. 
Let $x' \in J \cap B$ be a vector maximizing $f(x')$ subject to ${\rm dist}_{B} (x, x') \le 2$. 
Then, $f(x') - f(x) \ge \frac{2}{\|x^* - x \|_1} \left( f( x^*) - f(x) \right)$.  
\end{lemma}

\begin{proof}
Since $x, x^* \in J$, there exists a $2$-step decomposition $\{p_1, \dots , p_\ell\}$ of $x^*-x$ and $g_i \in \mathbb{R}$ for $i \in [\ell]$ 
such that 
$f(x + \sum_{i \in I} p_i) \ge f(x) + \sum_{i \in I} g_i$ for any $I \subseteq [\ell]$ and 
$f(x^*) = f(x) + \sum_{i \in [\ell]} g_i$.
For $i \in [\ell]$, we define $w_i = g_i - \frac{f(x^*) - f(x) }{\ell} + \varepsilon$, where $\varepsilon$ is a sufficiently small positive number. 
Then, by the same argument as the proof of Lemma~\ref{lem:updateratio}, 
we obtain $f(x') - f(x) \ge \frac{2}{\|x^* - x \|_1} \left( f( x^*) - f(x) \right)$. 
\end{proof}

This implies that
the global optimality is guaranteed by the local optimality as follows.

\begin{corollary}
\label{cor:weightlocalopt}
In an instance of {\sc Valuated Jump System Intersection} with (C2'), 
a feasible solution $x \in J \cap B$ maximizes $f(x)$ if and only if 
$f(x) \ge f(x')$ for any $x' \in J \cap B$ with ${\rm dist}_{B} (x, x') \le 2$. 
\end{corollary}

We are now ready to prove Theorem~\ref{thm:mainweighted}. 

\begin{proof}[Proof of Theorem~\ref{thm:mainweighted}]
By the same argument as Theorem~\ref{thm:main}, in which we use Lemma~\ref{lem:updateratioweighted} instead of Lemma~\ref{lem:updateratio}, 
we see that Algorithm~\ref{alg:02} terminates by executing  $O(B_{\rm size} \log ({\rm OPT} - f(x_0)) )$ iterations. 
Thus, the running time is polynomial in $\sum_{v \in V} \sum_{\alpha \in B(v)} \log ( | \alpha | + 1) + \max_{x \in J} \log ( |f(x)| + 1)$. 
The optimality of the obtained solution is guaranteed by Corollary~\ref{cor:weightlocalopt}. 
\end{proof}

\begin{remark}
\label{rem:mjump}
Functions with (SBO-M-JUMP) form a subclass of  
M-concave functions on constant parity jump systems studied in the context of discrete convex analysis~\cite{KMTJumpM,MinamikawaShioura,murotaDCA,MurotaJumpM}. 
For $J \subseteq \mathbb{Z}^V$, a function $f \colon J \to \mathbb{Z}$ is called 
an {\em M-concave function on a constant parity jump system}~\cite{MurotaJumpM} if it satisfies the following exchange axiom. 

\begin{description}
\item[(M-JUMP)]
For any $x, y \in J$ and for any $(x, y)$-step $s$, 
there exists an $(x+s, y)$-step $t$ such that 
$f(x+s+t) + f(y-s-t) \ge f(x) + f(y)$. 
\end{description}

We can see that  (SBO-M-JUMP) implies (M-JUMP) as follows. 
For $x, y \in J$, suppose that there exist a $2$-step decomposition  $\{p_1, \dots , p_\ell \}$ of $y-x$
and $g_i \in \mathbb{R}$ for $i \in [\ell]$ satisfying the conditions in (SBO-M-JUMP). 
For any $(x, y)$-step $s$, there exists an $(x+s, y)$-step $t$ such that $s+t= p_i$ for some $i \in [\ell]$. 
Such $t$ satisfies the conditions in (M-JUMP), because 
\begin{align*}
f(x+s+t) + f(y-s-t) &= f(x+p_i) + f \Big( x+ \sum_{j \in [\ell] \setminus \{i\}} p_j \Big) \\
&\ge (f(x) + g_i) +  \Big( f(x) + \sum_{j \in [\ell] \setminus \{i\}} g_j  \Big)  = f(x) + f(y). 
\end{align*}
\end{remark}

\section{Weighted General Factor Problem}
\label{sec:weightedGF}

It was shown by Dudycz and Paluch~\cite{DudyczP17} that 
the edge-weighted variant of the optimal general factor problem can also be solved in polynomial time if each $B(v)$ has no gap of length more than one. 
Formally, in the {\em weighted optimal general factor problem}, 
given a graph $G=(V, E)$, an edge weight $w(e) \in \mathbb{Z}$ for $e \in E$,  and a set $B(v) \subseteq \mathbb Z$ of integers for each $v \in V$, 
we seek for a $B$-factor $F \subseteq E$ that maximizes its total weight $\sum_{e \in F} w(e)$, where we denote $w(F) := \sum_{e \in F} w(e)$. 
Their algorithm consists of local improvement steps used in  Algorithms~\ref{alg:01} and~\ref{alg:02} and a scaling technique. 

In what follows in this section, 
we show that the polynomial solvability of the weighted optimal general factor problem is derived from Theorem~\ref{thm:mainweighted}.

\begin{theorem}[Dudycz and Paluch~\cite{DudyczP17}]
The weighted optimal general factor problem can be solved in polynomial time if each $B(v)$ has no gap of length more than one. 
\end{theorem}

\begin{proof}
Let $G=(V, E)$, $w$, and  $B$ be an instance of the weighted optimal general factor problem
such that each $B(v)$ has no gap of length more than one. 
Let $J := \{d_F \mid F \subseteq E \}$, and define 
$f \colon J \to \mathbb{Z}$ by $f(x) := \max \{ w(F) \mid d_F=x,\, F \subseteq E  \} $ for $x \in J$. 

We now show (C1'), (C2'), and (C3') in Theorem~\ref{thm:mainweighted}. 
Since an edge set $F_0 \subseteq E$ with $d_{F_0} \in B$ can be found  in polynomial time by the algorithm of Cornu\'ejols~\cite{CORNUEJOLS1988185} (if it exists), 
we obtain $x_0 := d_{F_0}$ satisfying the condition in (C1'). 
The subproblem in (C3') is to find an $(a, b)$-factor with parity constraints that maximizes the total edge weight, 
which can be solved in polynomial time; see~\cite[Section 35]{lexbook}. 
To see (C2'), 
for $x, y \in J$, let $M, N \subseteq E$ be edge sets such that $d_M = x$, $d_N = y$, $w(M) = f(x)$, and $w(N) = f(y)$.
As in Example~\ref{ex:02}, 
the symmetric difference of $M$ and $N$ can be decomposed into alternating paths $P_1, \dots , P_\ell$ and alternating cycles 
such that $\{d_{N\cap P_i} - d_{M \cap P_i} \mid i \in [\ell] \}$ is a $2$-step decomposition of $y-x$.
For  $i \in [\ell]$, let 
$p_i := d_{N\cap P_i} - d_{M \cap P_i}$ and $g_i := w(N\cap P_i) - w(M\cap P_i)$. 
For $I \subseteq [\ell]$,  let $F_I  \subseteq E$ be the symmetric difference of $M$ and $\bigcup_{i \in [I]} P_i$. 
Then, since $d_{F_I} = x + \sum_{i \in I} p_i$ and $w(F_I) = f(x) + \sum_{i \in I} g_i$, we obtain $f(x + \sum_{i \in I} p_i) \ge f(x) + \sum_{i \in I} g_i$. 
This shows (C2'). 

By Theorem~\ref{thm:mainweighted}, 
we can find $x^* \in J \cap B$ maximizing $f(x^*)$ in polynomial time. 
Furthermore, an edge set $F^* \subseteq E$ satisfying 
$w(F^*) = f(x^*)$ and $d_{F^*}=x^*$ can also be found  in polynomial time
by a weighted $b$-factor algorithm. 
By definition, such $F^*$ is an optimal solution of the weighted optimal general factor problem. 
\end{proof}

\section{Concluding Remarks}
\label{sec:conclusion}

In this paper, we have revealed that (SBO-JUMP) is a key property to obtain a polynomial time-algorithm for {\sc Jump System Intersection}, 
which is an abstract form of the optimal general factor problem. 
By using this abstraction, we have obtained a simpler correctness proof for the polynomial solvability of the optimal general factor problem.  
We have also extended the results to the valuated case.  

There are some possible directions for future research. 
First,  it is nice if we obtain more examples of jump systems satisfying  (SBO-JUMP) other than Examples~\ref{ex:01}--\ref{ex:02}, 
because they support the importance of this class of jump systems. 
Second, it is a natural open problem whether we can obtain a strongly polynomial-time algorithm for the weighted general factor problem. 
Finally, it is interesting to find a new property of $J$ other than (SBO-JUMP) that enables us to design a different polynomial-time algorithm.

\section*{Acknowledgements}
The author thanks Kenjiro Takazawa for his helpful comments. 
This work is supported by JSPS KAKENHI grant numbers 20K11692 and 20H05795, Japan.

\end{document}